\documentclass[letterpaper, 10 pt, conference]{ieeeconf}  

\IEEEoverridecommandlockouts                              

\usepackage[utf8]{inputenc}
\usepackage{amsmath}
\usepackage{amssymb}
\usepackage{amsfonts}
\usepackage{bm}

\usepackage{amsthm}
\usepackage{overpic}
\usepackage{lipsum}
\let\labelindent\relax
\usepackage{enumitem}
\usepackage{cases}
\usepackage{subcaption}
\usepackage{graphicx}
\usepackage{ulem,color}
\usepackage{cite}

\usepackage{algorithm, algorithmicx}
\usepackage{algpseudocode}

\algnewcommand{\Inputs}[1]{%
  \State \textbf{Inputs:}
  \Statex \hspace*{\algorithmicindent}\parbox[t]{.8\linewidth}{\raggedright #1}
}
\algnewcommand{\Initialize}[1]{%
  \State \textbf{Initialize:}
  \Statex \hspace*{\algorithmicindent}\parbox[t]{.8\linewidth}{\raggedright #1}
}

\DeclareMathOperator*{\argmax}{arg\,max}

{
    \theoremstyle{plain}
    \newtheorem{assumption}{Assumption}
}

\DeclareCaptionLabelFormat{lc}{\MakeLowercase{#1}~#2}

\newtheorem{theorem}{Theorem}

\newtheorem{definition}{Definition}
\newtheorem{proposition}{Proposition}
\newtheorem{example}{Example}[]

\title{\LARGE \bf
Hamilton's Rule for Enabling Altruism in Multi-Agent Systems 
}

\author{Brooks A. Butler and Magnus Egerstedt$^*$
\thanks{*Brooks A. Butler and Magnus Egerstedt are with the Department of Electrical Engineering and Computer Science at the University of California, Irvine. Emails: bbutler2@uci.edu and magnus@uci.edu.
This research was supported {in part} by an appointment to the Intelligence Community Postdoctoral Research Fellowship Program at the University of California, Irvine administered by Oak Ridge Institute for Science and Education (ORISE) through an interagency agreement between the U.S. Department of Energy and the Office of the Director of National Intelligence (ODNI). This work is also sponsored in part by the Air Force Research Laboratory, Munitions Directorate (RWTA), Eglin AFB, FL (Award No. FA8651-24-2-0001). Opinions, findings and conclusions, or recommendations are those of the authors and do not necessarily reflect the views of the sponsoring agencies.}
}

\begin{document}

\IEEEaftertitletext{\vspace{-1.5\baselineskip}}

\maketitle
\thispagestyle{empty}
\pagestyle{empty}

\begin{abstract}
This paper explores the application of Hamilton’s rule to altruistic decision-making in multi-agent systems. Inspired by biological altruism, we introduce a framework that evaluates when individual agents should incur costs to benefit their neighbors. By adapting Hamilton’s rule, we define agent ``fitness" in terms of task productivity rather than genetic survival. We formalize altruistic decision-making through a graph-based model of multi-agent interactions and propose a solution using collaborative control Lyapunov functions. The approach ensures that altruistic behaviors contribute to the collective goal-reaching efficiency of the system. We illustrate this framework on a multi-agent way-point navigation problem, where we show through simulation how agent importance levels influence altruistic decision-making, leading to improved coordination in navigation tasks.
\end{abstract}

\section{Introduction}
One of reasons why multi-robot systems have been proposed for deployment in risky or uncertain environments is {their} inherent robustness to failures of individual robots \cite{hazon2008redundancy,yan2013survey}. If one robot breaks down, there are still N-1 robots remaining who can carry on with the mission. This is in contrast to single robot deployments, where such malfunctions are mission-ending. But, if such robustness is a desired feature, one can take this one step further and ask if there are situations in which individual robots should sacrifice themselves for the good of the team, or at least {volunteer to} perform high-risk maneuvers?

This idea of incurring a purposeful cost (such as risk) for individual robots is reminiscent of the idea of altruism, whereby individuals or organisms perform acts that are costly to themselves to benefit a receiver organism or organisms. In ecology, Hamilton's rule \cite{hamilton1963evolution} has been formulated to encode when an altruistic act is beneficial from {the} vantage point of ``genetic fitness". Although genetic fitness might not be particularly relevant in a multi-robot setting, the underpinning idea of trading off costs and benefits across members of the team relative to the importance of the individual tasks is meaningful. In this paper, we formalize what Hamilton's rule might look like for robotics and showcase its operation in a collaborative task assignment scenario.

Risk-aware planning for teams of robots has been approached previously in several applications, with examples including human-robot collaborative teams \cite{jiang2022risk,kwon2020humans,li2023risk}, where robots can collaborate with humans to asses and react to risky environments, as well as risk-aware path planning \cite{xiao2020robot}, target tracking \cite{mayya2022adaptive}, and decision-making \cite{palomares2016collaborative}. In contrast to these past approaches, the novelty of this paper is the explicit coupling between purposeful risk-taking, or deteriorating performance, and algorithmic altruism for evaluating when decisions that are detrimental to the individual agent should be made towards the benefit of related agents. Additionally, approaches to altruistic decision making in multi-agent systems have been investigated in scenarios where agents might consider balancing an individual cost with a social/group cost when agents share a common goal \cite{ruadulescu2020multi,lu2021swarm,zamora1998learning,toghi2022social}. However, in this paper, we leverage the framework inspired by Hamilton's rule that weighs the costs/benefits of altruistic actions in the context of how agent {goals} relate to each other, even in the face of individualized goals, where a greater degree of {goal-}relatedness {will} influence the level of altruism an agent is willing to enact to benefit another.

The outline of this paper is as follows. In Section~\ref{sec:hamilton_rule}, we introduce Hamilton's rule from the perspective of ecology and propose how this framework can be applied to autonomous multi-agent. We then formalize our altruistic decision-making framework in the context of using a graph to model multi-agent systems with coupled dynamics in Section~\ref{sec:problem_formulation}. We then propose a solution to this formulation using collaborative control via the definition of collaborative control Lyapunov functions (CCLFs) in Section~\ref{sec:altruism_via_collab}, and present a result on the total goal-reaching of a collaborating system with respect to the relative importance of each agent. Finally, we illustrate our altruistic control framework on a simulated multi-agent way-point navigation problem and provide brief conclusions in Sections~\ref{sec:simulations} and \ref{sec:conclusion}, respectively.

\section{Hamilton's Rule for Autonomous Systems} \label{sec:hamilton_rule}

In ecology, \textit{genetic fitness} is used to quantify the reproductive success of a given organism \cite{orr2009fitness,reed2003correlation}. Perhaps most recognizably, this notion of fitness is used to describe evolutionary phenomena such as natural selection through \textit{survival of the fittest}, which suggests that genes that increase the likelihood of survival in an individual organism are more likely to create copies of themselves through reproduction, therefore increasing their genetic dominance in a population~\cite{endler1986natural}. \textit{Inclusive fitness} contrasts with this purely individualistic model of genetic success by suggesting that organisms with the same, or similar, genes may take \textit{altruistic} actions that support the survival of each other, even at a cost to themselves, to theoretically enhance the genetic fitness of both the recipient of the act and the altruistic organism. Hamilton's rule \cite{hamilton1963evolution} underpins this theory of inclusive fitness by deriving a condition under which altruistic genes are likely to propagate throughout a given population,     
\begin{equation} \label{eq:hamiltons_rule}
    r_{ij} B_j(u_i) \geq C_i(u_i),
\end{equation}
where, typically, $C_i(u_i)$ is the reproductive cost to organism $i$ of a given choice $u_i$, $B_j(u_i)$ is the benefit of the choice $u_i$ to organism $j$, and $r_{ij}$ is the genetic relatedness between the to organisms, which, in the ecological setting, {means} the probability that $i$ and $j$ share the same genes.

Taking inspiration from these {ecological} principles, 
we are motivated to emulate emergent strategies observed in nature that reward cooperation in complex systems of independent organisms. In short, we hope to provide a framework that can replicate the benefits of altruistic decision-making in ecological systems, which raises the overall fitness of a group of related organisms, to increase the total fitness of related agents in autonomous multi-agent systems.
However, transferring principles of genetic fitness from biological to autonomous systems requires a shift in how we measure success for an autonomous agent versus a biological organism. Rather than using fitness to describe the \textit{fecundity} (i.e., reproductive rate) of an agent, we instead consider agent fitness as a measure of \textit{productivity} (i.e., task/goal completion rate \cite{nguyen2023mutualistic,nguyen2024scalable}). Under this conceptual re-framing, we examine the components of altruistic decision-making, as modeled by Hamilton's rule, in the context of multi-agent systems: namely, relating agent inputs to costs and benefits on the productivity of other agents (i.e., defining $B_j(u_i)$ and $C_i(u_i)$), and defining measures for relatedness with respect to agent tasks/goals (i.e., deriving $r_{ij}$).     

\subsection{Conditions for Altruism in Multi-Agent Systems}
To define conditions where altruistic decision-making may be beneficial for overall productivity, we first highlight some core properties that should be satisfied in the definition of the cost/benefit framework. First, note that, in the ecological setting, the input/choice $u_i$ in \eqref{eq:hamiltons_rule} is generally implied and typically binary (i.e., either organism~$i$ helps or does not help organism~$j$), where the units of $B_j(u_i)$ and $C_i(u_i)$ are usually given in reproductive value, or number of offspring. For example, we might consider the benefit of a lioness feeding her sister's cub that would otherwise starve at the cost of some food/resources that would otherwise go to her own healthy cub \cite{Lotha_2022,clutton2013social}. In this scenario, $u_i$ would be the choice by lioness~$i$ to feed the cub of lioness~$j$, where $B_j(u_i) = 1$ offspring for lioness~$j$ and the cost could be $C_i(u_i)=\frac{1}{4}$ offspring for lioness~$i$ (i.e., some loss in resources for the healthy cub that could reduce the chance of survival). Additionally, since lioness $j$'s offspring would share at most half of its genes with lioness $i$, the final piece of \eqref{eq:hamiltons_rule} is given as $r_{ij}=0.5$, making the decision $u_i$, in this case, beneficial for the inclusive fitness of both lionesses.

However, in most autonomous systems, the input $u_i$ is generally designed to be explicit with a range of inputs (e.g. some vector of real values, $u_i \in \mathcal{U}_i \subset \mathbb{R}^{M_i}$), where it is conceivable that one input could be marginally better than another, although both could be technically beneficial. For example, in a collision avoidance scenario, does agent~$i$ divert from its desired heading by a small or large degree to give the right-of-way to agent~$j$? Furthermore, the units of $B_j(u_i)$ and $C_i(u_i)$ for all agents must also be consistent with respect to the impact that the input $u_i$ has on neighboring agent productivity, where a definition of ``productivity units" may not be immediately obvious. 

By way of example, consider a multi-agent system of {planar} single integrator agents, where the position of each agent is given by $x_i \in \mathbb{R}_i^2$ and {the} agent dynamics by
\begin{equation}
    \dot{x}_i = u_i.
\end{equation}
For $n$ agents, where $[n] = \{1, \dots, n\}$ denotes the set of agent indices, let each agent $i \in [n]$ be given a waypoint $x_i^* \in \mathbb{R}^2$, where each agent's goal is to navigate from their current position to said waypoint.
One way we can evaluate the current productivity of a given agent is by the rate which they approach their given goal,
\begin{equation}\label{eq:productivity_single_int}
    P_i(x_i, u_i) = (x_i^* - x_i)^\top u_i.
\end{equation}
Thus, assuming an input $u_i \in \mathcal{U}_i \subset \mathbb{R}^2$ {with finite magnitude}, an agent will be ``maximally productive" when their dynamics point them directly towards the given waypoint $x_i^*$ (i.e., maximizing the inner product between $(x_i^* - x_i)$ and $u_i$), with any $u_i$ that decreases the distance from $x_i^*$ resulting in positive productivity, and any $u_i$ that increases the distance from $x_i^*$ resulting in negative productivity. In this sense, we could use $P_i(x_i, u_i)$ as a candidate for a productivity cost function for agent~$i$ with respect to a given input $u_i$
\begin{equation}\label{eq:C_i_candidate_example}
    C_i(x_i, u_i) = -P_i(x_i, u_i).
\end{equation}

An essential component remains, however, to evaluate the effect of agent decisions on each other (i.e., $B_j(u_i)$), which is a mechanism, or constraint, that dictates agent behaviors with respect to their neighbors. In the current example, we can use collision avoidance, or relative agent spacing, as a coupling dynamic between agents. Since we are interested in how agents would \textit{react} to avoid potential collisions (rather than guaranteeing the collision avoidance itself) let us induce coupling dynamics between agents within a defined distance that encourages each agent to be spaced at least some distance $\Delta > 0$ apart. Define the neighborhood of agent~$i$ to be the following set of agents,
\begin{equation}
    \mathcal{N}_i = \{ j \in [n] \setminus \{i\}: \Vert x_i - x_j \Vert \leq \Delta \},
\end{equation}
where we want agents within $\Delta$ of each other to naturally distance themselves away from each other. In this sense, one could interpret agents within $\Delta$ of each other to be ``uncomfortable", but not necessarily unsafe. Thus, we can add an agent repulsion term to the dynamics of all agents
\begin{equation} \label{eq:uncomfortable_dyn}
    \dot{x}_i = u_i + u_{\mathcal{N}_i}(x_i, x_{\mathcal{N}_i}),
\end{equation}
where, {for example,}
\begin{equation} \label{eq:repelling_feild}
    u_{\mathcal{N}_i}(x_i, x_{\mathcal{N}_i}) = \gamma \sum_{j \in \mathcal{N}_i}  \frac{x_i - x_j}{\Vert x_i - x_j \Vert^2},
\end{equation}
{which} simulates a repelling potential field for all agents within $\Delta$ of each other, with scalar parameter $\gamma > 0$ to dictate the repulsion strength, where closer agents to each other have a stronger desire to move apart. While the choice of \eqref{eq:repelling_feild} is not unique \cite{zhang2010cooperative,sun2017collision}, as any similar repelling potential field mechanism would suffice for this example, we choose an inverse square law for simplicity, as it may also be viewed analogously to the repelling effect that same-signed charges have on each other (although not identically, since $u_{\mathcal{N}_i}(x_i, x_{\mathcal{N}_i})$ is not a force, but a velocity command).


Under these dynamics, the effect of neighboring agents on the goal-reaching of agent~$i$ is given by
\begin{equation}
    Q_i(x_i, x_{\mathcal{N}_i}) = \gamma (x_i^* - x_i)^\top \left( \sum_{j \in \mathcal{N}_i}  \frac{x_i - x_j}{\Vert x_i - x_j \Vert^2} \right),
\end{equation}
where $Q_i(x_i, x_{\mathcal{N}_i}) > 0$ implies that agent~$i$ is moved closer to $x_i^*$ and $Q_i(x_i, x_{\mathcal{N}_i}) \leq 0$ implies that agent~$i$ is not moved closer to $x_i^*$, due to the sum effect of all neighbors, respectively. Therefore, we can compute the effect that a single neighbor $j \in \mathcal{N}_i$ has on the goal-reaching of agent~$i$ as
\begin{equation}\label{eq:dQdxdt}
    \frac{\partial Q_i}{\partial x_j} \frac{\partial x_j}{ \partial t} = \frac{\partial Q_i}{\partial x_j} (u_j + u_{\mathcal{N}_j}(x_j, x_{\mathcal{N}_j}) ).
\end{equation}

Consider when there are only two agents $i$ and $j$ whose initial conditions and given waypoints cause the agents to cross within $\Vert x_i - x_j \Vert \leq \Delta$ of each other. 
In this case, we could pose a candidate for the benefit that agent~$i$ has on the goal-reaching of agent~$j$ as
\begin{equation}\label{eq:B_j_candidate_example}
    B_j(x_i, x_j, u_i) = \frac{\partial Q_j(x_i, x_j)}{\partial x_i} (u_i + u_{\mathcal{N}_i}(x_i, x_j)). 
\end{equation}
We should note that, even in this simplified example, the unit consistency of $C_i$ and $B_i$ is somewhat problematic, where the units of \eqref{eq:C_i_candidate_example} are clearly in the first-order dynamics of agent~$i$, whereas the units of \eqref{eq:B_j_candidate_example} are in the second-order dynamics for agent~$j$, that is, \eqref{eq:B_j_candidate_example} computes how $Q_j$ changes (which is computed with respect to the first-order coupling dynamics) under the dynamics of agent~$i$. However, noting this mismatch in units, we can still illustrate altruism in this example using the proposed candidates for $C_i$ and $B_j$, and will address this challenge more thoroughly in Sections~\ref{sec:problem_formulation} and \ref{sec:altruism_via_collab}. 
Thus, without considering relatedness between agents $i$ and $j$ (or by assuming $r_{ij} = 1$), by \eqref{eq:hamiltons_rule}, \eqref{eq:C_i_candidate_example}, and \eqref{eq:B_j_candidate_example} we have that any input $u_i$ that satisfies

\small
\begin{equation} \label{eq:alt_cond_example}
    \begin{aligned}
       \left( (x_i^* - x_i)^\top + \frac{\partial Q_j(x_i, x_j)}{\partial x_i} \right) u_i  \geq -\frac{\partial Q_j(x_i, x_j)}{\partial x_i}u_{\mathcal{N}_i}(x_i, x_j)
    \end{aligned}    
\end{equation}
\normalsize
will be at least equally beneficial to the goal-reaching of agent~$j$ as it is {detrimental} to agent~$i$. 

In Figure~\ref{fig:red_eq_blue}, we show an example of this condition applied to two agents, where $r_{ij}=r_{ji}=1$, whose given goals are to swap positions, with each agent attempting to maximize \eqref{eq:productivity_single_int} subject to the dynamics in \eqref{eq:uncomfortable_dyn} and the altruism condition \eqref{eq:alt_cond_example}. In this case, we see that a scenario that would normally result in a deadlock between both agents, with the $\argmax_{u_i \in \mathcal{U}_i} P_i(x_i, u_i)  = - u_{\mathcal{N}_i}(x_i, x_{\mathcal{N}_i})$, instead resolves the conflict by filtering individualistic actions according to~\eqref{eq:alt_cond_example}.
However, in order to fully implement the notion of altruistic decision-making in this example, we also require a method for determining agent relatedness in a productivity-centered context, which is described in the following section.

\subsection{Relatedness in Multi-Agent Systems}
Changing the units of agent success from fecundity to productivity also suggests a shift in agent-relatedness to a goal-centered perspective. 
One candidate for deriving agent relatedness is by considering the relationship of relative task importance between agents, computed as
\begin{equation}
    r_{ij} = \frac{w_j}{w_i},
\end{equation}
where $w_i,w_j \in \mathbb{R}_{>0}$ are used to denote a positive scalar importance value of the task/goal for agents~$i$ and $j$, respectively. 
This importance level could represent a constant priority assigned to the completion of agent~$i$'s task/goal; however, one might also imagine a time-varying importance $w_i(t)$, where importance levels are dictated by a dynamic environment with updating information. In either case, for simplicity, we assume that these values are given and reflect the relationship of agent goals with respect to each other.

Additionally, note that, in contrast with the concept of genetic relatedness which constrains the relatedness of agents to the set $r_{ij} \in [0,1]$, relative importance allows for $r_{ij} \in \mathbb{R}_{>0}$, where $r_{ij} = 1$ implies the goals of agent~$i$ and $j$ are equally important, $0 < r_{ij} < 1$ implies agent~$i$'s goal is more important relative to agent~$j$'s goal, and vice versa for $r_{ij} > 1$. 
We implement this notion of relative importance in our agent swapping example in Figures~\ref{fig:red_g_b} and \ref{fig:red_g_b_offset}, where, when $w_{\text{blue}} < w_{\text{red}}$, we see that the blue agent actually reverses its movement (despite its goal) to accommodate for the red agent until they are both able to navigate around each other, with the blue agent taking the larger detour.  


\begin{figure}
    \centering
    \begin{subfigure}[b]{0.32\columnwidth}
         \includegraphics[width=\textwidth]{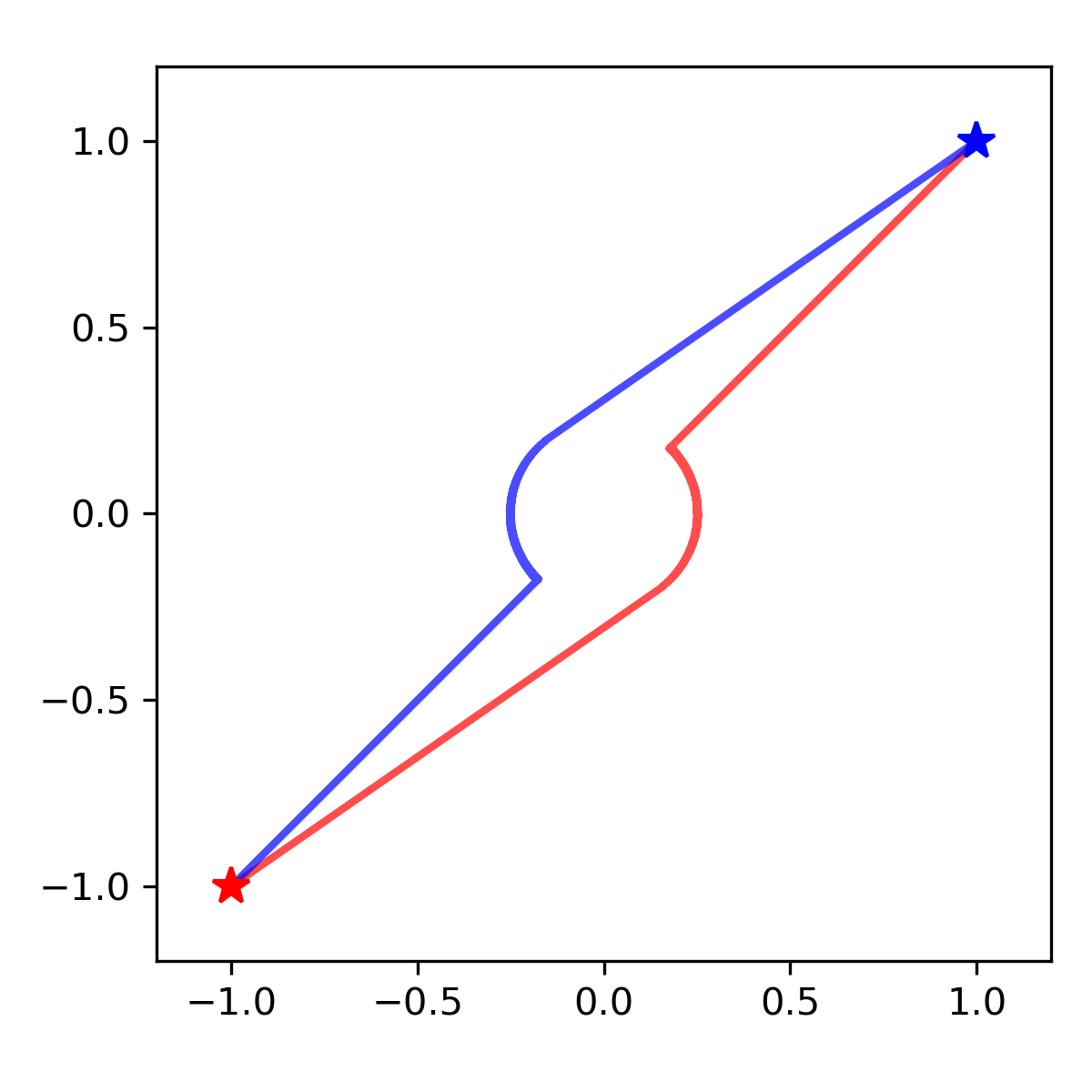}
         \caption{$w_{\text{blue}} = w_{\text{red}}$}
         \label{fig:red_eq_blue}
    \end{subfigure}
    \begin{subfigure}[b]{0.32\columnwidth}
         \includegraphics[width=\textwidth]{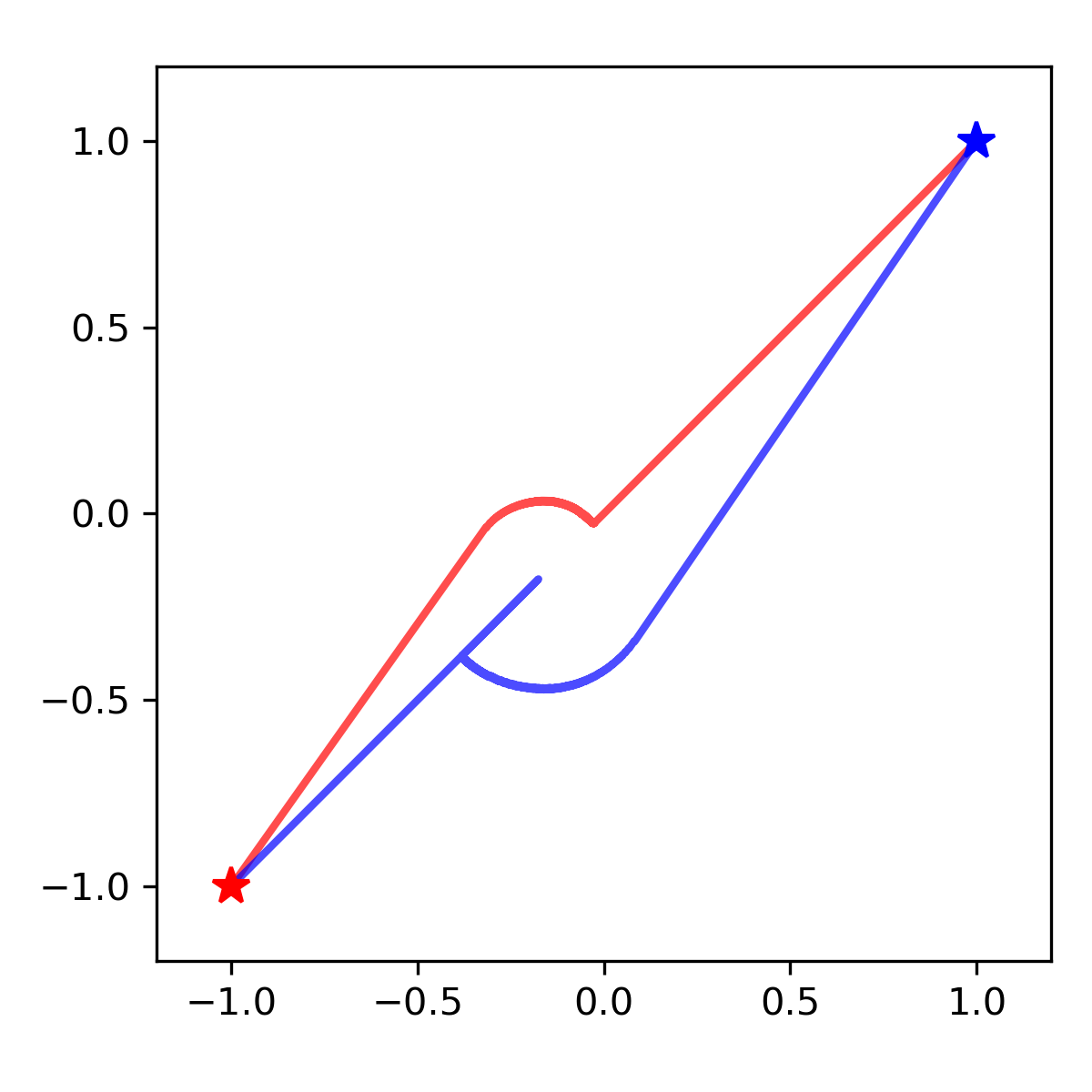}
         \caption{$w_{\text{blue}} < w_{\text{red}}$}
         \label{fig:red_g_b}
    \end{subfigure}
    \begin{subfigure}[b]{0.32\columnwidth}
         \includegraphics[width=\textwidth]{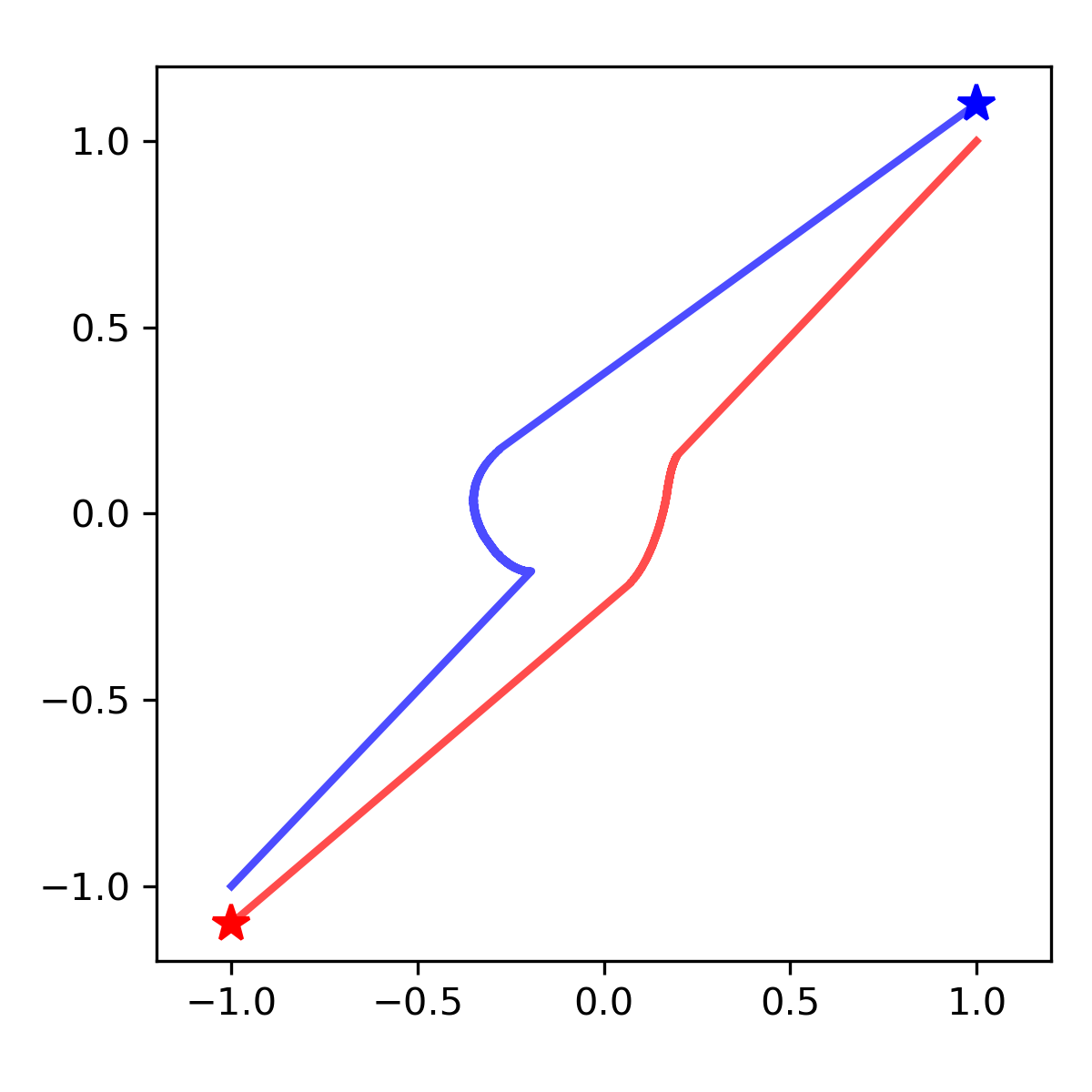}
         \caption{$w_{\text{blue}} < w_{\text{red}}$}
         \label{fig:red_g_b_offset}
    \end{subfigure}
    \caption{
    An example of the altruism condition from \eqref{eq:alt_cond_example} applied to two agents with goals to swap positions. We plot the trajectories of both agents, where in Figure~\ref{fig:red_eq_blue} the importance values are equal and in Figure~\ref{fig:red_g_b} the importance of the blue agent is set to be less than the red agent, with Figure~\ref{fig:red_g_b_offset} showing an example with agents swap positions slightly offset.
    }\label{fig:hamilitons_rule_example}
\end{figure}

\section{Problem Formulation} \label{sec:problem_formulation}
We now seek to formalize altruistic decision-making for non-linear multi-agent systems. As noted in Section~\ref{sec:hamilton_rule}, a key component of determining the effect of agents' decisions on each other is considering the mechanism by which agents affect each other's dynamics. 
If we consider agents to be distinct entities with some effect on the dynamics of neighboring agents, one way we can model these effects is by using a graph, where nodes in the graph represent individual agents and edges between agents indicate a dynamic effect between said agents. Note that there is a distinction between \textit{{graph dynamics}} or a \textit{networked dynamic system} (i.e., a graph that describes the dynamic effects of agents on each other) and a \textit{communication graph} (i.e., a graph that represents which agents may communicate between each other), where the latter is commonly used throughout the literature to facilitate coordination and/or cooperation in multi-agent systems \cite{li2014designing,you2011network,wen2012consensus}.
We define a networked dynamic system using a graph $\mathcal{G} = (\mathcal{V}, \mathcal{E})$, where $\mathcal{V}$ is the set of $n = \vert \mathcal{V} \vert$ nodes and $\mathcal{E} \subseteq \mathcal{V}\times \mathcal{V}$ is the set of edges.

A defining characteristic of a graph $\mathcal{G}$ is whether it is \textit{directed} or \textit{undirected}, where a graph is undirected if an edge from node~$i$ to $j$ also implies an edge from node~$j$ to $i$. While communication graphs are typically assumed to be undirected, this may not always be the case for networked dynamic systems. For example, consider the case where a team of bicyclists form a platoon to create a draft that optimizes effort expenditures for trailing teammates, where a bicyclist has a dynamic effect on all trailing teammates, while trailing teammates do not {necessarily} have a dynamic effect on those in the lead.
In this sense, let $\mathcal{N}_i^+$ be the set of all neighbors $j \in [n] \setminus \{i\}$ that have an interaction with the dynamics node~$i$, where 
\begin{equation}
    \mathcal{N}_i^+ = \{j \in [n]\setminus \{i\}: (i,j) \in \mathcal{E} \}.
\end{equation}
Similarly, all nodes $j \in [n] \setminus \{i\}$ whose dynamics are affected by node~$i$ are given by
\vspace{-1ex}
\begin{equation}
    \mathcal{N}_i^- = \{j \in [n]\setminus \{i\}: (j,i) \in \mathcal{E} \},
\end{equation}
with the complete set of neighboring nodes given by
\begin{equation}
    \mathcal{N}_i = \mathcal{N}_i^+ \cup \mathcal{N}_i^-.
\end{equation}

\noindent
Further, we define the state vector for each node as $x_i \in \mathbb{R}^{N_i}$, with $N = \sum_{i \in [n]} N_i$ being the state dimension of the entire system, $N_i^+ = \sum_{j \in \mathcal{N}_i^+} N_j$ the combined dimension of incoming neighbor states, and $x_{\mathcal{N}_i^+} \in \mathbb{R}^{N_i^+}$ denoting the combined state vector of all incoming neighbors.  Then, for each agent~$i \in [n]$, we can describe its state dynamics, which are potentially nonlinear, time-invariant, and control-affine, as
\begin{equation} \label{eq:net_dyn_sys}
    \dot{x}_i = f_i(x_i, x_{\mathcal{N}_i^+}) + g_i(x_i) u_i,
\end{equation}
where $f_i:\mathbb{R}^{N_i + N_i^+} \rightarrow \mathbb{R}^{N_i}$ and $g_i:  \mathbb{R}^{N_i} \rightarrow \mathbb{R}^{N_i} \times \mathbb{R}^{M_i}$ are locally Lipschitz for all $i \in [n]$, and $u_i \in \mathcal{U}_i \subset \mathbb{R}^{M_i}$. 
For notational compactness, given a node~$i \in [n]$, we collect the 1-hop neighborhood state as $\mathbf{x}_i = (x_i, x_{\mathcal{N}_i})$ and the 2-hop neighborhood state as $\mathbf{x}_i^+ = (x_i, x_{\mathcal{N}_i}, x_{\mathcal{N}_j}:\forall j \in \mathcal{N}_i)$. Although the form of \eqref{eq:net_dyn_sys} is typically used to model systems with natural, or uncontrolled, coupling dynamics (e.g., power systems, networked epidemics, physically coupled systems, etc.), it can also be useful for modeling altruism in systems with induced coupling behavior (e.g., collision avoidance, formation control, coverage control, etc.), where $u_i$ may be viewed as a modification to the induced behavior  as  
\begin{equation}\label{eq:induced_net_dyn}
    \dot{x}_i = f_i(x_i) + g_i(x_i)(u_{\mathcal{N}_i}(\mathbf{x}_i)+u_i).
\end{equation}
We can then rewrite \eqref{eq:induced_net_dyn} identically to \eqref{eq:net_dyn_sys} as 
\begin{equation} \label{eq:induced_coupled_dynamics}
    \dot{x}_i = \bar{f}_i(\mathbf{x}_i) + g_i(x_i)u_i
\end{equation}
where
\begin{equation}
    \bar{f}_i(\mathbf{x}_i) = f_i(x_i) + g_i(x_i)u_{\mathcal{N}_i}(\mathbf{x}_i).
\end{equation}

In this way, we contextualize altruistic behavior around the effects of coupling protocols, both natural and induced, which are the mechanisms that agents use to affect one another. A benefit of this formulation is that it allows us to think separately from the design of a specific coupling behavior, either natural or induced, between agents; however, a drawback to this approach is that it requires the induced coupling control law $u_{\mathcal{N}_i}(\mathbf{x}_i)$ to be differentiable in $\mathbf{x}_i$ for all agents. This need for a differentiable $u_{\mathcal{N}_i}(\mathbf{x}_i)$ arises from needing to consider how the dynamics of a given agent affect the goal-reaching of another through the coupling dynamics, where, in our example from Section~\ref{sec:hamilton_rule}, this manifests in the computation of $\frac{\partial Q_i}{\partial x_j} \frac{\partial x_j}{ \partial t}$ in \eqref{eq:dQdxdt}.

Thus, using the framing of networked dynamic systems, we propose a framework for deriving a condition under which altruistic actions are beneficial/advantageous for the entire multi-agent system given some measure of productivity or goal completion for all agents. Since we incorporate the networked dynamics into our formulation, this goal is described formally as a constrained control problem
\begin{equation} \label{eq:problem_statement}
    \begin{aligned}
        \min_{u_i} & \quad \sum_{i \in [n]} \Vert u_i \Vert \\
        \text{s.t.} & \quad \dot{x}_i = f_i(\mathbf{x}_i) + g_i(x_i)u_i  \\
        & \quad  C_i(\mathbf{x}_i, u_i) - \sum_{j \in \mathcal{N}_i^-} r_{ij} B_j(\mathbf{x}_j,\mathbf{x}_i, u_i) \leq 0; \forall i \in [n], 
    \end{aligned}
\end{equation}
where $C_i$ describes the cost of a given input to the goal-reaching of agent~$i$ and $B_j$ describes the benefit of $u_i$ to the goal-reaching of agent~$j$, with $r_{ij}$ given as the relative importance of agent~$j$'s task with respect to the task of agent~$i$. 
Additionally, as noted in Section~\ref{sec:hamilton_rule}, we see, under the framing of networked dynamics, that the effect of neighboring agent \textit{dynamics} (including the actions of neighbors $u_j$ for $j \in \mathcal{N}_i$) do not appear explicitly in the first-order dynamics of agent~$i$, but rather in the second-order dynamics. Thus, to maintain unit consistency, both $C_i$ and $B_j$ should describe the cost and benefit of actions $u_i$ in the second-order dynamics of each agent, which will also require the 2-hop neighborhood state $\mathbf{x}_i^+$ for each agent.

\section{Altruism Through Collaborative Control} \label{sec:altruism_via_collab}
In this section, we propose a method for constructing the cost/benefit functions $C_i(\mathbf{x}_i, u_i)$ and $B_j(\mathbf{x}_j,\mathbf{x}_i, u_i)$ necessary for solving \eqref{eq:problem_statement}. However, we must first define a context for evaluating the costs of agent actions. One method we can use to quantify the cost of input $u_i$ agent~$i$ is by evaluating the effect input $u_i$ has on the instantaneous progress of agents toward their given goal. In the following, we show how this notion of instantaneous progression can be viewed naturally through the lens of control Lyapunov functions (CLFs) \cite{sontag1983lyapunov,freeman1996clf,ogren2001control,ames2014rapidly} and how, by extending CLFs into the second-order dynamics of a networked dynamic system, we can describe the effect of agent's control inputs on neighbor's goals. 

\subsection{Control Lyapunov Functions}
For each agent~$i \in [n]$, let $V_i:\mathbb{R}^{N_i} \rightarrow \mathbb{R}_{\geq 0}$ describe a candidate control Lyapunov function, where the mission goal, defined with respect to each agent, is at $V_i(x_i) = 0$
We can verify that $V_i$ is a valid control Lyapunov function for a given agent~$i$ if the following definition holds. Let $D^r$ denote the set of functions $r$-times continuously differentiable in all arguments.
\begin{definition}[Control Lyapunov Function \cite{ogren2001control}] \label{def:clf}
    Given $f_i, g_i \in  D^{\infty}$ and $x_0 \in \mathbb{R}^{N_i}$. Then there exists a feedback law $u_i = \alpha_i(x_i)$ (smooth everywhere except at $x_0$, where it is continuous), which globally, asymptotically stabilizes $\dot{x}_i = f_i(\mathbf{x}_i) + g_i(x_i)u_i$ to $x_0$ if and only if there exists a control Lyapunov function, i.e. a smooth and positive definite function $V_i(x_i,x_0)$ (zero only at $x_0$) such that the following holds:
    \begin{equation}
        \frac{\partial V_i}{\partial x_i}g_i(x_i) = 0 \implies \frac{\partial V_i}{\partial x_i}f_i(\mathbf{x}_i) < 0
    \end{equation}
\end{definition}

Given that the Lie derivative of the function $V:\mathbb{R}^{N} \rightarrow \mathbb{R}$ with respect to the vector field generated by $f:\mathbb{R}^N \rightarrow \mathbb{R}^N$ is defined as
\begin{equation}
    \mathcal{L}_f V(x) = \frac{\partial V(x)}{\partial x} f(x),
\end{equation}
we can express the time-derivative of $V_i$, which indicates whether agent~$i$ approaches its goal, as
\begin{equation}
    \dot{V}_i(\mathbf{x}_i, u_i) = \mathcal{L}_{f_i}V_i(\mathbf{x}_i)+ \mathcal{L}_{g_i}V_i(x_i)u_i.
\end{equation}
Furthermore, we can guarantee that agent~$i$ reaches its goal at an exponential rate if the following condition holds
\begin{equation}\label{eq:CLF_sigma}
    \dot{V}_i(\mathbf{x}_i, u_i) + \sigma_i(V_i(x_i))  \leq 0
\end{equation}
holds, where $\sigma_i(\cdot)$ is a class~$\mathcal{K}$ function that determines the rate of convergence \cite{ames2014rapidly}. 

\subsection{Collaborative Control Lyapunov Functions}
Using the CLF condition from \eqref{eq:CLF_sigma}, we can interpret $\dot{V}_i + \sigma_iV_i \leq 0$ as agent~$i$ approaching its goal (i.e., a negative cost, or a positive benefit) and $\dot{V}_i + \sigma_i V_i > 0$ as agent~$i$ moving away from its goal (i.e., a positive cost). Thus, we could potentially use the CLF condition as a candidate for the cost of an action $u_i$ with respect to the productivity (or goal approaching) of agent~$i$,
\begin{equation*}
    C_i(\mathbf{x}_i, u_i) \approx \dot{V}_i(\mathbf{x}_i, u_i) + \sigma_i(V_i(x_i)).
\end{equation*}
However, using this candidate raises the issue of how to then design $B_j$ to reflect how $u_i$ also affects agent~$j$ using the same units (i.e., within the first order dynamics of agent~$j$). In fact, by turning our attention to $\dot{V}_i$, we see that only the \textit{states} of neighboring agents in $\mathcal{N}_i^+$ shows up in this expression. Therefore, we can not evaluate the impact of neighbor's inputs $u_j$ on $\dot{V}_i$ in the first-order dynamics.

However, intuitively, neighbors' choices clearly can have an impact on the instantaneous goals of agent~$i$. We can see this effect explicitly by examining the second-order dynamics of $V_i$
\vspace{-2ex}

\footnotesize
\begin{equation} \label{eq:ddot_Vi}
    \begin{aligned}
        \ddot{V}_i(\mathbf{x}_i^+, u_i, \dot{u}_i, u_{\mathcal{N}_i^+}) &= \sum_{j \in \mathcal{N}_i^+} \big[ \mathcal{L}_{f_j} \mathcal{L}_{f_i}  V_i(\mathbf{x}_i, \mathbf{x}_j) +  \mathcal{L}_{g_j} \mathcal{L}_{f_i}  V_i(\mathbf{x}_i) u_j \big] \\
         &
         + \mathcal{L}^2_{f_i} V_i(\mathbf{x}_i) +  u_i^\top \mathcal{L}^2_{g_i} V_i(x_i) u_i + \mathcal{L}_{g_i}V_i(x_i)\dot{u}_i \\ 
         &
         + \big(\mathcal{L}_{f_i} \mathcal{L}_{g_i}V_i(\mathbf{x}_i)^\top +  \mathcal{L}_{g_i} \mathcal{L}_{f_i}V_i(\mathbf{x}_i)\big)u_i,
    \end{aligned}
\end{equation}
\normalsize
where we define high-order Lie derivatives with respect to the same vector field $f$ with a recursive formula, where $k>1$, as
\begin{equation}
    \mathcal{L}^k_f h(x) = \frac{\partial \mathcal{L}^{k-1}_f V(x)}{\partial x} f(x)
\end{equation}
and we compute the Lie derivative of $V$ along the vector field generated by $f$ and then along the vector field generated by $g$ as
\begin{equation}
    \mathcal{L}_g \mathcal{L}_f V(x) = \frac{\partial}{\partial x}\left(\frac{\partial V(x)}{\partial x} f(x)\right)g(x).
\end{equation}

Using the second-order dynamics of $V_i$, we can now construct a high-order CLF condition that describes the convergence of agent~$i$ to its goal as a function of both $u_i$ and $u_j$ for $j \in \mathcal{N}_i^+$. We do so by defining a series of functions, similarly to the construction of a \textit{high-order control barrier function} \cite{tan2021high,xiao2021high,butler2024collaborative,butler2024resilience}

\small
\begin{equation} \label{eq:HOCLF_funcs}
    \begin{aligned}
        \phi^0_i(x_i) &= V_i(x_i) \\ 
        \phi^1_i(\mathbf{x}_i, u_i) &= \dot{\phi}^{0}_i(\mathbf{x}_i, u_i) + \sigma^1_i(\phi^{0}_i(x_i)) \\
        \phi^2_i(\mathbf{x}_i^+, u_i, \dot{u}_i, u_{\mathcal{N}_i^+}) &= \dot{\phi}^{1}_i(\mathbf{x}_i^+, u_i, \dot{u}_i, u_{\mathcal{N}_i^+}) + \sigma^2_i(\phi^{1}_i(\mathbf{x}_i, u_i)),
    \end{aligned}
\end{equation}
\normalsize

\noindent
where $\sigma^1_i(\cdot)$ and $\sigma^2_i(\cdot)$ are class-$\mathcal{K}$ functions. Note that $\phi^1_i(\mathbf{x}_i, u_i)$ is identical to the left-hand side of \eqref{eq:CLF_sigma}, and that $\phi^2_i(\mathbf{x}_i^+, u_i, \dot{u}_i, u_{\mathcal{N}_i^+})$ contains \eqref{eq:ddot_Vi}, which includes both the control inputs of neighbors in $\mathcal{N}_i^+$ and $\dot{u}_i$. We make the following simplifying assumption.   

\begin{assumption} \label{assume:u_dot_func_u}
    Let $\dot{u}_i := \tilde{d}(u_i) + \epsilon$, where $\tilde{d}(u_i): \mathbb{R}^{M_i} \rightarrow \mathbb{R}^{M_i}$ is a locally-Lipschitz approximation of $\dot{u}_i$, given a choice of $u_i \in \mathbb{R}^{M_i}$, with some finite approximation error $\epsilon \in \mathbb{R}^{M_i}$, where $\Vert \epsilon \Vert$ is sufficiently small.
\end{assumption}
Since implementing more sophisticated methods for handling/estimating $\dot{u}$ in high-order certificate-based control methods is an area of active research \cite{ong2024rectified} (and therefore out of the scope of this paper), we assume that $\dot{u}_i$ can be estimated given an input $u_i$ with some reasonably small error. Thus, using \eqref{eq:HOCLF_funcs} and Assumption~\ref{assume:u_dot_func_u}, we are prepared to define collaborative control Lyapunov functions as follows.

\begin{definition}[Collaborative Control Lyapunov Function] \label{def:cclf}
    For a given agent~$i$, under Assumption~\ref{assume:u_dot_func_u}, $V_i$ is a collaborative control Lyapunov function (CCLF) if there exists positive constants $c_1,c_2 > 0$ and class-$\mathcal{K}$ functions $\sigma^1_i(\cdot), \sigma^2_i(\cdot)$ such that for all $\mathbf{x}_i^+$ the following inequalities hold
    \begin{align}
        c_1\Vert x_i \Vert^2 \leq V_i(x_i) \leq c_2\Vert x_i \Vert^2, \label{eq:cclf_exp_stab_cond} \\
        \inf_{u_i, u_{\mathcal{N}_i^+}} \phi^2_i(\mathbf{x}_i^+, u_i, u_{\mathcal{N}_i}, \epsilon) \leq 0 \label{eq:collab_conv_cond}
    \end{align}
\end{definition}
\noindent
Using this definition, we obtain the following expected result.
\begin{proposition}
    If $V_i$ is a CCLF, then there exists a set of control inputs $u_i(t)$ and $u_{\mathcal{N}_i}(t)$ that globally exponentially stabilize $V_i$ at the origin.
\end{proposition}
\begin{proof}
    If $V_i$ is a CCLF, then by Definition~\ref{def:cclf} and \eqref{eq:HOCLF_funcs} we have that there exist class-$\mathcal{K}$ functions $\sigma_i^1$ and $\sigma_i^2$ such that 
    \begin{equation*}
        \begin{aligned}
            \inf_{u_i, u_{\mathcal{N}_i^+}} \dot{\phi}^{1}_i(\mathbf{x}_i^+, u_i, u_{\mathcal{N}_i^+}, \epsilon) + \sigma^2_i(\phi^{1}_i(\mathbf{x}_i, u_i)) \leq 0,
        \end{aligned}
    \end{equation*}
    for all $\mathbf{x}_i^+$. Thus, if \eqref{eq:collab_conv_cond} holds for all $\mathbf{x}_i^+$, then for any $\mathbf{x}_i, u_i$ such that $\phi^{1}_i(\mathbf{x}_i, u_i) = 0$ there must also exist inputs $u_{\mathcal{N}_i^+}$ such that $\dot{\phi}^{1}_i(\mathbf{x}_i^+, u_i, u_{\mathcal{N}_i^+}, \epsilon) \leq 0$, making $\phi^{1}_i(\mathbf{x}_i, u_i)$ monotonically decreasing. Therefore, by Definition~\ref{def:clf}, $V_i$ must also be a CLF, thereby making $V_i$ globally asymptotically stable, with exponential stability following directly from \eqref{eq:cclf_exp_stab_cond} in Definition~\ref{def:cclf}.  
\end{proof}

Note that the definition of a CCLF is not inherently useful for stabilizing $V_i$, as satisfying the simpler condition of the control Lyapunov function in Definition~\ref{def:clf} would be sufficient to prove global exponential stability of $V_i$. However, what the CCLF does provide is a condition by which we can simultaneously evaluate the contribution of both $u_i$ and $u_{\mathcal{N}_i^+}$ to the stabilization of $V_i$, which we can now use to construct the necessary cost functions for evaluating the benefit of altruistic actions of agents with respect to goals of other agents.     

\subsection{Collaborative Conditions for Altruism}
Given the definition of a CCLF for agent~$i$, we can examine the components of the critical condition in \eqref{eq:collab_conv_cond} by organizing the terms according those terms relating to $u_i$ and $u_{\mathcal{N}_i^+}$ as follows
\begin{equation}
    \phi^2_i(\mathbf{x}_i^+, u_i, u_{\mathcal{N}_i}) = \sum_{j \in \mathcal{N}_i^+} a_{ij}(\mathbf{x}_i, \mathbf{x}_j, u_j) + b_i(\mathbf{x}_i, u_i, \dot{u}_i), 
\end{equation}
where
\begin{equation}
    a_{ij}(\mathbf{x}_i, \mathbf{x}_j, u_j) = \mathcal{L}_{f_j} \mathcal{L}_{f_i} V_i(\mathbf{x}_i, \mathbf{x}_j) + \mathcal{L}_{g_j}\mathcal{L}_{f_i} V_i(\mathbf{x}_i)u_j
\end{equation}
and
\begin{equation} \label{eq:capability_node_i}
    \begin{aligned}
        b_i(\mathbf{x}_i,u_i,\dot{u}_i) &= u_i^\top \mathcal{L}^2_{g_i} V_i(x_i) u_i + q_i(\mathbf{x}_i) u_i \\ 
        & \quad + \mathcal{L}_{g_i} V_i(x_i)\dot{u}_i  + c_i(\mathbf{x}_i).
    \end{aligned}
\end{equation}
Although choosing an appropriate class-$\mathcal{K}$ function is non-trivial in the design of CLFs, to simplify the derivation of \eqref{eq:capability_node_i} in this paper, we choose the set of linear class-K functions.

\begin{assumption}\label{assumne:linear_class_K}
    Let $\sigma_i^k(z) := \sigma_i^k z$, where $z \in \mathbb{R}^{N_i}$ and $\sigma_i^k \in \mathbb{R}_{\geq 0}$.
\end{assumption}

\noindent
Under Assumption~\ref{assumne:linear_class_K}, we derive $q_i(\mathbf{x}_i)$ and $c_i(\mathbf{x}_i)$ as

\small
\begin{equation}
    q_i(\mathbf{x}_i) = \mathcal{L}_{f_i} \mathcal{L}_{g_i}V_i(\mathbf{x}_i)^\top + \mathcal{L}_{g_i} \mathcal{L}_{f_i}V_i(\mathbf{x}_i) + (\sigma_i^{1} + \sigma_i^{2}) \mathcal{L}_{g_i}V_i(x_i)
\end{equation}
\normalsize
and
\begin{equation}
    \begin{aligned}
        c_i(\mathbf{x}_i) &= \mathcal{L}^2_{f_i} V_i(\mathbf{x}_i) 
        + (\sigma_i^{1} + \sigma_i^{2}) \mathcal{L}_{f_i}V_i(\mathbf{x}_i) + \sigma_i^{1}\sigma_i^{2} V_i(x_i).
    \end{aligned}
\end{equation}
This collection of notation helps with separating terms relating to the effect of agent~$i$ on reaching its goals versus the impact of neighbors on agent~$i$ reaching its goal, where we can express the collaborative condition in \eqref{eq:collab_conv_cond} identically as
\begin{equation}
    \sum_{j \in \mathcal{N}_i^+} a_{ij}(\mathbf{x}_i,\mathbf{x}_j,u_j) + b_i(\mathbf{x}_i, u_i, \dot{u}_i) \leq 0.
\end{equation}

Considering the effect that agent~$i$ has on the goal of agent~$j$ in the context of the CCLF, we can now formulate an altruistic rule similarly to Hamilton's rule, where we now have the ability to describe the benefit of input $u_i$ to the goal-reaching of agent~$j$ in the second-order dynamics as
\begin{equation}
    B_{j}(\mathbf{x}_j,\mathbf{x}_i,u_i) := - a_{ji}(\mathbf{x}_j, \mathbf{x}_i, u_i) 
\end{equation}
We negate this because a lower value of $a_{ji}(\mathbf{x}_j, \mathbf{x}_i, u_i)$ is more beneficial to the goal-reaching of agent~$j$ (i.e., more negative values denote faster convergence of $V_j$ to the origin). Naturally, we use the remaining terms of the CCLF to comprise the cost of input $u_i$ to agent~$i$ as
\begin{equation}
    C_i(\mathbf{x}_i,u_i,\dot{u}_i) := b_i(\mathbf{x}_i,u_i,\dot{u}_i),
\end{equation}
where we do not negate this term because a more positive value does indeed translate to a higher cost to agent~$i$. Thus, for two agents, we now have a Hamilton's rule-like condition, where any input $u_i$ that satisfies
\begin{equation}
   b_i(\mathbf{x}_i,u_i,\dot{u}_i) + r_{ij} a_{ji}(\mathbf{x}_j, \mathbf{x}_i,u_i) \leq 0  
\end{equation}
will increase the inclusive productivity of both agents according the importance $w_i, w_j$ assigned each task, which gives the relatedness term $r_{ij} = \frac{w_j}{w_i}$. Considering an arbitrary number of agents, we can account for the total effect of input $u_i$ on all neighbors $j \in \mathcal{N}_i^-$ as
\begin{equation}\label{eq:hamiltons_rule_CCLF}
    b_i(\mathbf{x}_i,u_i,\dot{u}_i) + \sum_{j \in \mathcal{N}_i^-} r_{ij} a_{ji}(\mathbf{x}_j,\mathbf{x}_i,u_i) \leq 0.
\end{equation}

We first examine the effect of \eqref{eq:hamiltons_rule_CCLF} on the total convergence of a multi-agent system under the following assumption.

\begin{assumption} \label{assume:undirected_graph}
    Let $\mathcal{N}_i^+ = \mathcal{N}_i^-$, for all $i \in [n]$.
\end{assumption}
\noindent
In other words, for our analysis, we assume that the {graph dynamics} for the system is undirected, which yields the following result on the total convergence of the system under the notion of altruistic decision-making.

\begin{theorem} \label{thm:sum_phi_0}
    If Assumption~\ref{assume:undirected_graph} and \eqref{eq:hamiltons_rule_CCLF} hold for all $i \in  [n]$, then
    \begin{equation}
        \sum_{i \in [n]} w_i \phi^2_i(\mathbf{x}_i^+, u_i, u_{\mathcal{N}_i}) \leq 0.
    \end{equation}
\end{theorem}
\begin{proof}
    For a given agent $i \in [n]$, adding $\sum_{j \in \mathcal{N}_i^+} a_{ij}(\mathbf{x}_i, \mathbf{x}_j, u_j)$ to \eqref{eq:hamiltons_rule_CCLF} yields
    \begin{equation*}
        \begin{aligned}  
            &\sum_{j \in \mathcal{N}_i^+} a_{ij}(\mathbf{x}_i,\mathbf{x}_j,u_j) - \sum_{j \in \mathcal{N}_i^-} r_{ij} a_{ji}(\mathbf{x}_j, \mathbf{x}_i, u_i) \\
            &\geq 
            b_i(\mathbf{x}_i,u_i,\dot{u}_i) + \sum_{j \in \mathcal{N}_i^+} a_{ij}(\mathbf{x}_i,\mathbf{x}_j,u_j) \\ 
            &= \phi^2_i(\mathbf{x}_i^+, u_i, u_{\mathcal{N}_i}). 
        \end{aligned}
    \end{equation*}
    Thus, we have that
    \small
    \begin{equation} \label{eq:sum_phi}
        \begin{aligned}  
            &
            \sum_{i \in [n]} w_i \phi^2_i(\mathbf{x}_i^+, u_i, u_{\mathcal{N}_i})   
            \\ 
            &
            \leq \sum_{i \in [n]} w_i \left( \sum_{j \in \mathcal{N}_i^+} a_{ij}(\mathbf{x}_i,\mathbf{x}_j,u_j) - \sum_{j \in \mathcal{N}_i^-} r_{ij} a_{ji}(\mathbf{x}_j, \mathbf{x}_i, u_i) \right).
        \end{aligned}
    \end{equation}
    \normalsize
    Therefore, under Assumption~\ref{assume:undirected_graph}, and using the fact that $r_{ij} = w_j/w_i$, we can rewrite the right-hand side of \eqref{eq:sum_phi} as
    \begin{equation*}
        \begin{aligned}
            \sum_{i \in [n]} \left( \sum_{j \in \mathcal{N}_i} w_i a_{ij}(\mathbf{x}_i,\mathbf{x}_j,u_j) -  w_j a_{ji}(\mathbf{x}_j, \mathbf{x}_i, u_i) \right) = 0,
        \end{aligned}
    \end{equation*}
    where the equality holds since, in the sum of terms for all agents, each weighted term $w_i a_{ij}(\mathbf{x}_i,\mathbf{x}_j,u_j)$ is canceled out by a corresponding negative term. Thus, we have that $\sum_{i \in [n]} w_i \phi^2_i(\mathbf{x}_i^+, u_i, u_{\mathcal{N}_i}) \leq 0$.   
\end{proof}

In other words, Theorem~\ref{thm:sum_phi_0} states that, under the altruistic rule defined by \eqref{eq:hamiltons_rule_CCLF}, the \textit{weighted total goal-reaching} of the system will be convergent (i.e., any agent movement away from their goal will be counteracted by another agent moving toward their goal, as weighted by their relative importance).

\section{Simulations} \label{sec:simulations}

In this section, we revisit our example from Section~\ref{sec:hamilton_rule} and apply the altruistic framework proposed in Section~\ref{sec:altruism_via_collab}. To fully realize a computable expression for \eqref{eq:hamiltons_rule_CCLF} for this example, we first observe that 

\begin{equation}\label{eq:Lgi2_eq_0}
    \mathcal{L}^2_{g_i} V_i(x_i) = \mathbf{0}^{M_i \times M_i}, \forall i \in [n].
\end{equation}
This property holds for any systems where $g_i$ is constant (i.e., robots that implement control uniformly regardless of their position in a physical space). Thus, for agents with dynamics defined by \eqref{eq:uncomfortable_dyn}, any input $u_i$ that satisfies
\begin{equation} \label{eq:altruism_cond_u_i}
    \begin{aligned}
        \left( \sum_{j \in \mathcal{N}_i^-} r_{ij} \mathcal{L}_{g_i}\mathcal{L}_{f_j} V_j(\mathbf{x}_j) + q_i(\mathbf{x}_i) \right) u_i + \mathcal{L}_{g_i} V_i(x_i) \tilde{d}(u_i) \\ 
        \leq - \sum_{j \in \mathcal{N}_i^-} r_{ij} \mathcal{L}_{f_i} \mathcal{L}_{f_j} V_j(\mathbf{x}_j, \mathbf{x}_i) - c_i(\mathbf{x}_i) - \mathcal{L}_{g_i} V_i(x_i)\epsilon
    \end{aligned}
\end{equation}
also satisfies \eqref{eq:hamiltons_rule_CCLF}.
Further, if the estimation for $\dot{u}_i$ is affine in $u_i$ (e.g., $\tilde{d}(u_i) = \frac{u_{0} - u_i}{dt}$), then \eqref{eq:altruism_cond_u_i} becomes a linear constraint with respect to $u_i$ and solutions to \eqref{eq:problem_statement} can be found efficiently using a quadratic program.

We define a CLF candidate function for each agent as 
\begin{equation}
    V_i(x_i) = \frac{1}{2} \Vert x_i^* - x_i \Vert^2,
\end{equation}
where $x_i^*$ is the goal-point for agent~$i$. At runtime, each agent computes a nominal control action according to the CLF condition in~\eqref{eq:CLF_sigma}, which is then filtered according to the altruism condition in~\eqref{eq:altruism_cond_u_i}. In Figure~\ref{fig:hamilitons_rule_example_8agent}, we show the trajectories of a simulated 8-agent system in two cases, where in Figure~\ref{fig:all_eq} we set $w_i=1$ for all agents and in Figure~\ref{fig:blue_g_orange_g_all} we increase the importance of the blue and orange agents to $w_{\text{blue}}=10^6$ and $w_{\text{orange}}=10^3$, respectively. In the case of greater importance assigned to the blue and orange agents in Figure~\ref{fig:blue_g_orange_g_all}, we see that all other agents modify their moments to not hinder the progression of the blue agent towards its goal. Furthermore, we show in Figure~\ref{fig:phi2} that Theorem~\ref{thm:sum_phi_0} holds even while some agents are moved away from their goal in reaction to more agents with greater assigned importance. 

\begin{figure}
    \centering
    \begin{subfigure}[b]{0.45\columnwidth}
         \includegraphics[width=\textwidth]{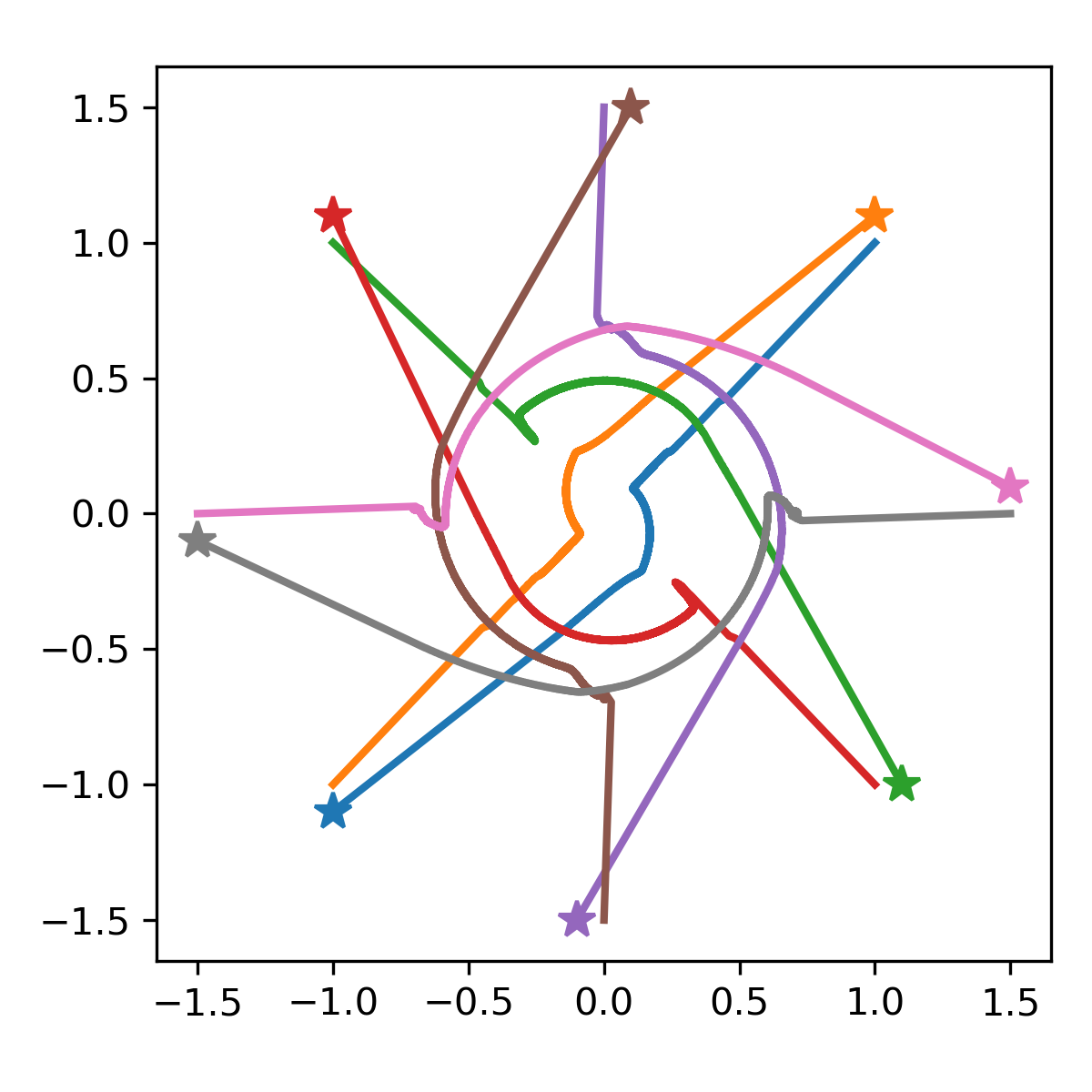}
         \caption{$w_{i} = w_{j} = 1, \forall i,j$}
         \label{fig:all_eq}
    \end{subfigure}
    \begin{subfigure}[b]{0.45\columnwidth}
         \includegraphics[width=\textwidth]{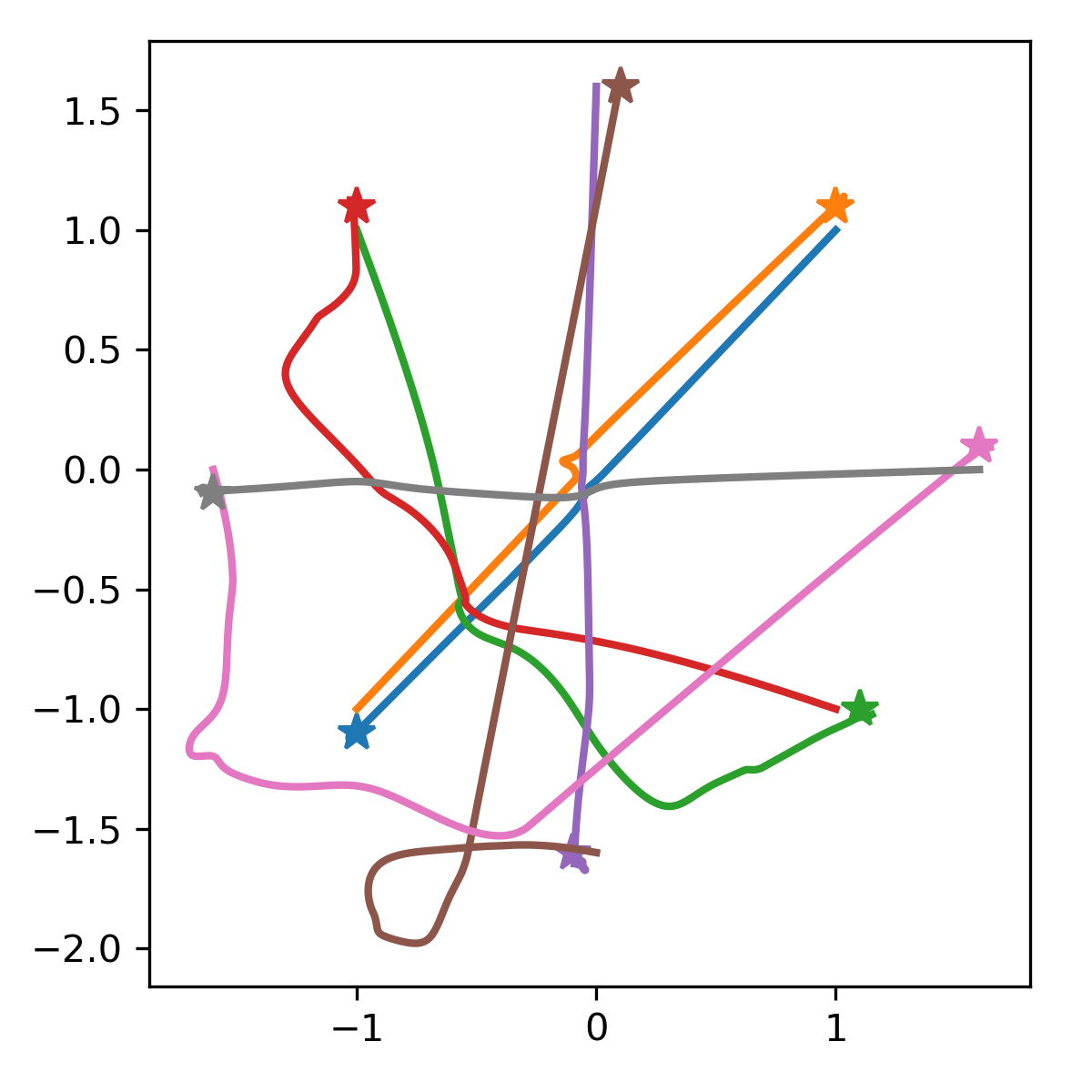}
         \caption{$w_{\text{blue}} \gg w_{\text{orange}} \gg 1$}
         \label{fig:blue_g_orange_g_all}
    \end{subfigure}
    \caption{
    Two examples of trajectories for an 8-agent system where a star of the corresponding color shows each agent's goal, with agent dynamics defined by \eqref{eq:uncomfortable_dyn} subject to the altruism condition defined by \eqref{eq:hamiltons_rule_CCLF}. In Figure~\ref{fig:all_eq}, all agents are assigned equal importance values, whereas in Figure~\ref{fig:blue_g_orange_g_all} we assign the greatest importance to the blue and orange agents, where $w_{\text{blue}}=10^6$ and $w_{\text{orange}}=10^3$, with all other agent importance being set at $1$.
    }\label{fig:hamilitons_rule_example_8agent}
\end{figure}

\begin{figure}
    \centering
    \includegraphics[width=.95\columnwidth]{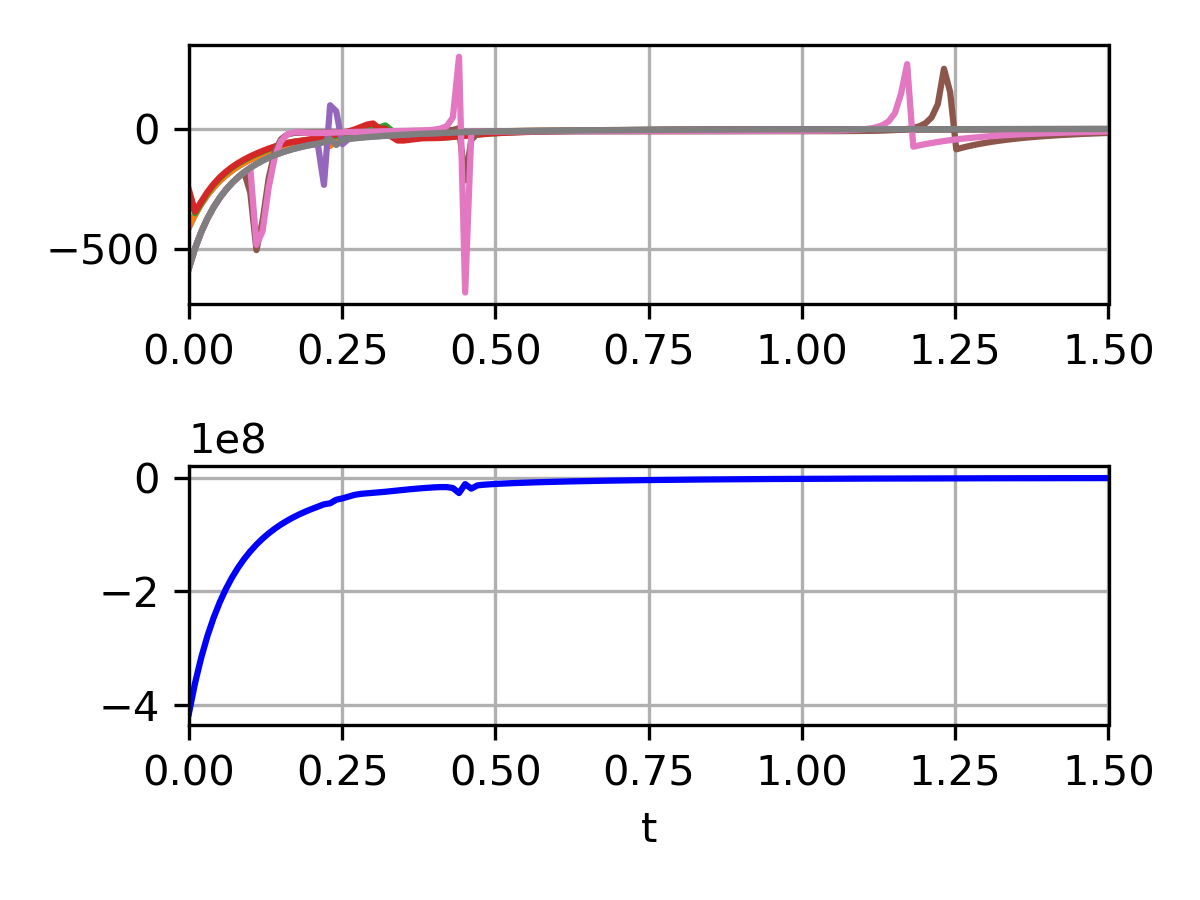}
    \caption{
    Plots of $\phi^2_i(\mathbf{x}_i^+, u_i, u_{\mathcal{N}_i})$ for all $i \in[n]$ (top) and $\sum_{i \in [n]} w_i \phi^2_i(\mathbf{x}_i^+, u_i, u_{\mathcal{N}_i})$ (bottom) over time for the 8-agent simulation shown in Figure~\ref{fig:blue_g_orange_g_all}. Note that, despite some agents moving away from their goal, Theorem~\ref{thm:sum_phi_0} still holds for all time.
    }
    \label{fig:phi2}
\end{figure}

\section{Conclusion} \label{sec:conclusion}
In the paper, we have investigated the application of Hamilton's rule to construct an ecology-inspired altruistic decision-making framework for autonomous multi-agent systems, where agent fitness is considered in the context of agent productivity. We have proposed a general framework that utilizes principles from certificate-based control methods to define collaborative control Lyapunov functions, which provide a metric by which each agent may compute the effect that it has on other agents' goal reaching, as defined by a Lyapunov-like condition. Further, we have proven that when agents adhere to altruistic decision-making principles, the weighted total goal-reaching of the system, as measured by the relative importance of each agent, is convergent. Some important directions for future work include extending our notion of agent goal-reach beyond the instantaneous effect of agent decision and instead considering predictive planning while adhering to altruistic conditions.

\section*{{Acknowledgment}}
{We would like to thank Scott Nivison for his helpful discussions and guidance in this work.}

\normalem
\bibliographystyle{IEEEtran}
\bibliography{references}

\end{document}